\documentclass[prodmode,acmtalg]{acmsmall} 

\usepackage[ruled]{algorithm2e}
\SetAlFnt{\small}
\SetAlCapFnt{\small}
\SetAlCapNameFnt{\small}
\SetAlCapHSkip{0pt}
\IncMargin{-\parindent}

\acmVolume{0}
\acmNumber{0}
\acmArticle{0}
\acmYear{0}
\acmMonth{0}

\usepackage{amsmath}
\usepackage{amsfonts}
\usepackage{amssymb}
\usepackage{MnSymbol}
\usepackage{ifsym}
\usepackage[utf8]{inputenc}
\usepackage{latexsym}
\usepackage{enumitem}
\usepackage{graphicx}
\usepackage{microtype}

\newcommand{\Osymbol}{{\mathcal O}}
\newcommand{\BO}[1]{\Osymbol\left(#1\right)}

\newcommand{\floor}[1]{\lfloor{#1}\rfloor}
\newcommand{\ceil}[1]{\lceil{#1}\rceil}
\newcommand{\E}[1]{\mathbb{E}\left[#1\right]}

\renewcommand{\Pr}[1]{\text{Pr}\left[#1\right]}

\doi{0000001.0000001}

\issn{1234-56789}

\begin{document}


\markboth{R. Pagh}{CoveringLSH: Locality-sensitive Hashing without False Negatives}

\title{CoveringLSH: Locality-sensitive Hashing without False Negatives}
\author{RASMUS PAGH
\affil{IT University of Copenhagen}}

\begin{abstract}
We consider a new construction of locality-sensitive hash functions for Hamming space that is \emph{covering} in the sense that is it guaranteed to produce a collision for every pair of vectors within a given radius $r$. 
The construction is \emph{efficient} in the sense that the expected number of hash collisions between vectors at distance~$cr$, for a given $c>1$, comes close to that of the best possible data independent LSH without the covering guarantee, namely, the seminal LSH construction of Indyk and Motwani (STOC '98).
The efficiency of the new construction essentially \emph{matches} their bound when the search radius is not too large --- e.g., when $cr = o(\log(n)/\log\log n)$, where $n$ is the number of points in the data set, and when $cr = \log(n)/k$ where $k$ is an integer constant.
In general, it differs by at most a factor $\ln(4)$ in the exponent of the time bounds. 
As a consequence, LSH-based similarity search in Hamming space can avoid the problem of false negatives at little or no cost in efficiency.
\end{abstract}

%
%


%


\acmformat{Rasmus Pagh, 2016. CoveringLSH: Locality-sensitive Hashing without False Negatives.}

\begin{bottomstuff}
The research leading to these results has received funding from the European Research Council under the European Union’s 7th Framework Programme (FP7/2007-2013) / ERC grant agreement no.~614331.
\end{bottomstuff}


\maketitle

\section{Introduction}

Similarity search in high dimensions has been a subject of intense research for the last decades in several research communities including theory of computation, databases, machine learning, and information retrieval.
In this paper we consider nearest neighbor search in Hamming space, where the task is to find a vector in a preprocessed set $S\subseteq \{0,1\}^d$ that has minimum Hamming distance to a query vector $y\in\{0,1\}^d$.

It is known that efficient data structures for this problem, i.e., whose query and preprocessing time does not increase exponentially with~$d$, would disprove the strong exponential time hypothesis~\cite{williams2005new,DBLP:conf/focs/AlmanW15}.
For this reason the algorithms community has studied the problem of finding a \emph{$c$-approximate} nearest neighbor, i.e., a point whose distance to $y$ is bounded by $c$ times the distance to a nearest neighbor, where $c > 1$ is a user-specified parameter.
If the exact nearest neighbor is sought, the approximation factor $c$ can be seen as a bound on the relative distance between the nearest and the second nearest neighbor.
All existing $c$-approximate nearest neighbor data structures that have been rigorously analyzed have one or more of the following drawbacks:
\begin{enumerate}
	\item Worst case query time linear in the number of points in the data set, or
	\item Worst case query time that grows exponentially with $d$, or
	\item Multiplicative space overhead that grows exponentially with $d$, or
	\item Lack of unconditional guarantee to return a nearest neighbor (or $c$-approximate nearest neighbor).
\end{enumerate}
Arguably, the data structures that come closest to overcoming these drawbacks are based on locality-sensitive hashing (LSH).
For many metrics, including the Hamming metric discussed in this paper, LSH yields sublinear query time (even for $d\gg \log n$) and space usage that is polynomial in $n$ and linear in the number of dimensions~\cite{Indyk1998,gionis1999similarity}.
If the approximation factor $c$ is larger than a certain constant (currently known to be at most~$3$) the space can even be made $\BO{nd}$, still with sublinear query time~\cite{panigrahy2006entropy,kapralov2015smooth,DBLP:journals/corr/Laarhoven15a}.

However, these methods come with a Monte Carlo-type guarantee:
A $c$-approximate nearest neighbor is returned only \emph{with high probability}, and there is no efficient way of detecting if the computed result is incorrect.
This means that they do not overcome the 4th drawback above.

\medskip

\paragraph{Contribution}
In this paper we investigate the possibility of Las Vegas-type guarantees for ($c$-approximate) nearest neighbor search in Hamming space.
Traditional LSH schemes pick the sequence of hash functions independently, 
which inherently implies that we can only hope for high probability bounds.
Extending and improving results by~\cite{greene1994multi} and~\cite{Arasu_VLDB06} we show that in Hamming space, by suitably correlating hash functions we can ``cover'' all possible positions of~$r$ differences and thus eliminate false negatives, while achieving performance bounds comparable to those of traditional LSH methods.
By known reductions~\cite{DBLP:conf/stoc/Indyk07} this implies Las Vegas-type guarantees also for $\ell_1$ and $\ell_2$ metrics.
Since our methods are based on combinatorial objects called \emph{coverings} we refer to the approach as \emph{CoveringLSH}.

\medskip

Let $||x-y||$ denote the Hamming distance between vectors $x$ and $y$.
Our results imply the following theorem on similarity search (specifically $c$-approximate near neighbor search) in a standard unit cost (word RAM) model:

\begin{theorem}\label{thm:RAM}
		Given $S\subseteq \{0,1\}^d$, $c>1$ and $r\in {\bf N}$, we can construct a data structure such that for $n=|S|$ and a value $f(n,r,c)$ bounded by
		 $$f(n,r,c) = \left\{ \begin{array}{ll}
 \BO{1} & \text{if } \log(n)/(cr)\in {\bf N}\\
 (\log n)^{\BO{1}} & \text{if } cr \leq \log(n)/(3\log\log n)\\
 \BO{\min\left(n^{0.4/c}\, r ,\, 2^r \right)} & \text{for all parameters}
 \end{array} \right. \enspace ,$$		 
	the following holds:
		\begin{itemize}
		\item On query  $y\in \{0,1\}^d$ the data structure is guaranteed to return $x\in S$ with $||x-y|| < cr$ if there exists $x'\in S$ with $||x'-y|| \leq r$.
		 \item The expected query time is $\BO{f(n,r,c)\, n^{1/c} (1+d/w)}$, where~$w$ is the word length.
		 \item The size of the data structure is $\BO{f(n,r,c)\, n^{1+1/c}\log n + nd}$ bits.
		\end{itemize}
\end{theorem}

\medskip

Our techniques, like traditional LSH, extend to efficiently solve other variants of similarity search.
For example, we can: 1) handle nearest neighbor search without knowing a bound on the distance to the nearest neighbor, 2) return \emph{all} near neighbors instead of just one, and 3) achieve high probability bounds on query time rather than just an expected time bound. 

When $f(n,r,c) = \BO{1}$ the performance of our data structure matches that of classical LSH with constant probability of a false negative~\cite{Indyk1998,gionis1999similarity}, so $f(n,r,c)$ is the multiplicative overhead compared to classical LSH.
In fact, \cite{o2014optimal} showed that the exponent of $1/c$ in query time is optimal for methods based on (data independent) LSH.

\subsection{Notation}

For a set $S$ and function $f$ we let $f(S) = \{ f(x) \; | \; x\in S\}$.
We use ${\bf 0}$ and ${\bf 1}$ to denote vectors of all 0s and 1s, respectively. 
For $x,y\in\{0,1\}^d$ we use $x\wedge y$ and $x\vee y$ to denote bit-wise conjunction and disjunction, respectively, and $x \oplus y$ to denote the bitwise exclusive-or.
Let $I(x) = \{ i \; | \; x_i = 1\}$.
We use $||x|| = |I(x)|$ to denote the Hamming weight of a vector~$x$, and 
$$||x-y|| = |I(x\oplus y)|$$
to denote the Hamming distance between $x$ and $y$.
For $S\subseteq \{0,1\}^d$ let $\Delta_S$ be an upper bound on the time required to produce a representation of the nozero entries of a vector in $S$ in a standard (word RAM) model~\cite{word-RAM}.
Observe that in general $\Delta_S$ depends on the representation of vectors (e.g., bit vectors for dense vectors, or sparse representations if $d$ is much larger than the largest Hamming weight).
For bit vectors we have $\Delta_S = \BO{1+d/w}$ if we assume the ability to count the number of 1s in a word in constant time\footnote{This is true on modern computers using the {\sc popcnt} instruction, and implementable with table lookups if $w=\BO{\log n}$. 
If only a minimal instruction set is available it is possible to get $\Delta_S = \BO{d/w + \log w}$ by a folklore recursive construction, see e.g.~\cite[Lemma 3.2]{hagerup2001deterministic}.}, and this is where the term $1+d/w$ in Theorem~\ref{thm:RAM} comes from.
We use ``$x \text{ mod } b$'' to refer to the integer in $\{0,\dots,b-1\}$ whose difference from $x$ is divisible by $b$.
Finally, let $\langle x,y\rangle $ denote $||x \wedge y||$, i.e., the dot product of $x$ and $y$.


\section{Background and related work}

Given $S\subseteq \{0,1\}^d$ the problem of searching for a vector in $S$ within Hamming distance~$r$ from a given query vector $y$ was introduced by Minsky and Papert as the \emph{approximate dictionary} problem~\cite{minsky1987perceptrons}.
The generalization to arbitrary spaces is now known as the \emph{near neighbor} problem (or sometimes as \emph{point location in balls}).
It is known that a solution to the approximate near neighbor problem for fixed $r$ (known before query time) implies a solution to the nearest neighbor problem with comparable performance~\cite{Indyk1998,HarPeled2012}.
In our case this is somewhat simpler to see, so we give the argument for completeness.
Two reductions are of interest, depending on the size of $d$.
If $d$ is small we can obtain a nearest neighbor data structure by having a data structure for every radius $r$, at a cost of factor $d$ in space and $\log d$ in query time.
Alternatively, if $d$ is large we can restrict the set of radii to the $\BO{\log(n)\log(d)}$ radii of the form $\lceil(1+1/\log n)^i\rceil < d$.
This decreases the approximation factor needed for the near neighbor data structures by a factor $1+1/\log n$, which can be done with no asymptotic cost in the data structures we consider.
For this reason, in the following we focus on the near neighbor problem in Hamming space where $r$ is assumed to be known when the data structure is created.

\subsection{Deterministic algorithms}

For simplicity we will restrict attention to the case $r\leq d/2$.
A baseline is the \emph{brute force} algorithm that looks up all $\binom{d}{r}$ bit vectors of Hamming distance at most~$r$ from $y$.
The time usage is at least $(d/r)^r$, assuming $r\leq d/2$, so this method is not attractive unless~$d^r$ is quite small.
The dependence on $d$ was reduced by~\cite{Cole:2004:DMI:1007352.1007374} who achieve query time $\BO{d+\log^r n}$ and space $\BO{nd + n \log^r n}$.
Again, because of the exponential dependence on~$r$ this method is interesting only for small values of $r$.

\subsection{Randomized filtering with false negatives}\label{sec:lsh}

In a seminal paper~\cite{Indyk1998}, Indyk and Motwani presented 
a randomized solution to the \emph{$c$-approximate} near neighbor problem where the search stops as soon as a vector within distance $cr$ from $y$ is found.
Their technique can also be used to solve the approximate dictionary problem, but the time will then depend on the number of points at distance between $r+1$ and $cr$ that we inspect.
Their data structure, like all LSH methods for Hamming space we consider in this paper, uses a set of functions from a \emph{Hamming projection} family:
\begin{equation}\label{eq:hamming-projection}
\mathcal{H}_\mathcal{A} = \{ x \mapsto x \wedge a \; | \; a\in \mathcal{A} \}
\end{equation}
where $\mathcal{A}\subseteq \{0,1\}^d$.
The vectors in $\mathcal{A}$ will be referred to as \emph{bit masks}.
Given a query $y$, the idea is to iterate through all functions $h\in \mathcal{H}_\mathcal{A}$ and identify collisions $h(x)=h(y)$ for $x\in S$, e.g.~using a hash table.
This procedure \emph{covers} a query $y$ if at least one collision is produced when there exists $x\in S$ with $||x-y||\leq r$, and it is \emph{efficient} if the number of hash function evaluations and collisions with $||x-y|| > cr$ is not too large.
The procedure can be thought of as a randomized \emph{filter} that attempts to catch data items of interest while filtering away data items that are not even close to being interesting.
The \emph{filtering efficiency} with respect to vectors $x$ and $y$ is the expected number of collisions $h(x)=h(y)$ summed over all functions  $h\in \mathcal{H}_\mathcal{A}$, with expectation taken over any randomness in the choice of~$\mathcal{A}$.
We can argue that without loss of generality it can be assumed that the filtering efficiency depends only on $||x-y||$ and not on the location of the differences.
To see this, using an idea from~\cite{Arasu_VLDB06}, consider replacing each $a\in \mathcal{A}$ by a vector $\pi(a)$ defined by $\pi(a)_i = a_{\pi(i)}$, where $\pi: \{1,\dots,d\}\rightarrow \{1,\dots,d\}$ is a random permutation used for all vectors in $\mathcal{A}$.
This does not affect distances, and means that collision probabilities will depend solely on $||x-y||$, $d$, and the Hamming weights of vectors in $\mathcal{A}$.

\medskip

\paragraph{Remark} If vectors in $\mathcal{A}$ are sparse it is beneficial to work with a sparse representation of the input and output of functions in $\mathcal{H}_\mathcal{A}$, and indeed this is what is done by Indyk and Motwani who consider functions that concatenate a suitable number of 1-bit samples from $x$.
However, we find it convenient to work with $d$-dimensional vectors, with the understanding that a sparse representation can be used if $d$ is large. $\triangle$

\medskip

\paragraph{Classical Hamming LSH}
Indyk and Motwani use a collection 
$$\mathcal{A}(R) = \{ a(v) \; | \; v\in R \},$$
where $R \subseteq \{1,\dots,d\}^k$ is a set of uniformly random and independent $k$-dimensional vectors.
Each vector $v$ encodes a sequence of $k$ samples from $\{1,\dots,d\}$, and $a(v)$ is the projection vector that selects the sampled bits.
That is, $a(v)_i = 1$ if and only if $v_j = i$ for some $j\in \{1,\dots,k\}$. 
By choosing $k$ appropriately we can achieve a trade-off that balances the size of $R$ (i.e., the number of hash functions) with the expected number of collisions at distance~$cr$.
It turns out that $|R| = \BO{n^{1/c} \log(1/\delta)}$ suffices to achieve collision probability $1-\delta$ at distance~$r$ while keeping the expected total number of collisions with ``far'' vectors (at distance $cr$ or more) linear in~$|R|$.

\medskip

\paragraph{Newer developments}
In a recent advance of~\cite{DBLP:conf/stoc/AndoniR15}, extending preliminary ideas from~\cite{andoni2014beyond}, it was shown how \emph{data dependent} LSH can achieve the same guarantee with a smaller family (having $n^{o(1)}$ space usage and evaluation time).
Specifically, it suffices to check collisions of $\BO{n^\rho \log(1/\delta)}$ hash values, where $\rho = \tfrac{1}{2c-1}+o(1)$.
We will not attempt to generalize the new method to the data dependent setting, though that is certainly an interesting possible extension.

In a surprising development, it was recently shown~\cite{DBLP:conf/focs/AlmanW15} that even with no approximation of distances ($c=1$) it is possible to obtain truly sublinear time per query if: 1) $d=\BO{\log n}$ \emph{and}, 2) we are concerned with the answers to a \emph{batch} of $n$ queries.

\subsection{Filtering methods without false negatives}\label{sec:nofalseneg}

The literature on filtering methods for Hamming distance that do not introduce false negatives, but still yield formal guarantees, is relatively small. 
As in section~\ref{sec:lsh} the previous results can be stated in the form of Hamming projection families (\ref{eq:hamming-projection}).
We consider constructions of sets $\mathcal{A}$ that ensure collision for every pair of vectors at distance at most~$r$, while at the same time achieving nontrivial filtering efficiency for larger distances.

Choosing error probability $\delta < 1/\binom{d}{r}$ in the construction of Indyk and Motwani, we see that there must exist a set $R^*$ of size $\BO{\log(1/\delta) n^{1/c}}$ that works for every choice of $r$ mismatching coordinates, i.e., ensures collision under some $h\in \mathcal{H}_{\mathcal{A}(R^*)}$ for all pairs of vectors within distance $r$. In particular we have $|\mathcal{A}(R^*)| = \BO{d n^{1/c}}$. However, this existence argument is of little help to design an algorithm, and hence we will be interested in \emph{explicit} constructions of LSH families without false negatives.\footnote{\cite{DBLP:conf/soda/Indyk00} sketched a way to \emph{verify} that a random family contains a colliding function for every pair of vectors within distance $r$, but unfortunately the construction is incorrect~\cite{IndykPersonalCommunication2015}.}

Kuzjurin has given such explicit constructions of ``covering'' vectors~\cite{kuzjurin2000explicit} but in general the bounds achieved are far from what is possible existentially~\cite{Kuzjurin:1995:DAG:204515.204520}.
Independently,~\cite{greene1994multi} linked the question of similarity search without false negatives to the Turán problem in extremal graph theory.
While optimal Turán numbers are not known in general, Greene et al.~construct a family $\mathcal{A}$ (based on \emph{corrector hypergraphs}) that will incur few collisions with \emph{random vectors}, i.e., vectors at distance about $d/2$ from the query point.\footnote{It appears that Theorem 3 of~\cite{greene1994multi} does not follow from the calculations of the paper --- a factor of about 4 is missing in the exponent of space and time bounds~\cite{ParnasPersonalCommunication}.}
\cite{gordon1995new} presented near-optimal coverings for certain parameters based on finite geometries --- in section~\ref{sec:small-radius} we will use their construction to achieve good data structures for small~$r$.

\cite{Arasu_VLDB06} give a construction that is able to achieve, for example, $o(1)$ filtering efficiency for approximation factor $c > 7.5$ with $|\mathcal{A}| = \BO{r^{2.39}}$.
Observe that there is no dependence on $d$ in these bounds, which is crucial for high-dimensional (sparse) data.
The technique of~\cite{Arasu_VLDB06} allows a range of trade-offs between $|\mathcal{A}|$ and the filtering efficiency, determined by parameters $n_1$ and $n_2$.
No theoretical analysis is made of how close to~$1$ the filtering efficiency can be made for a given $c$, but it seems difficult to significantly improve the constant~7.5 mentioned above.

Independently of the work of~\cite{Arasu_VLDB06}, ``lossless'' methods for near neighbor search have been studied in the contexts of approximate pattern matching~\cite{kucherov2005multiseed} and computer vision~\cite{norouzi2012fast}.
The analytical part of these papers differs from our setting by focusing on filtering efficiency for random vectors, which means that differences between a data vector and the query appear in random locations. 
In particular there is no need to permute the dimensions as described in section~\ref{sec:lsh}.
Such schemes aimed at random (or more generally ``high entropy'') data become efficient when there are few vectors within distance $r \log |S|$ of a query point.
Another variation of the scheme of~\cite{Arasu_VLDB06} recently appeared in~\cite{DBLP:journals/pvldb/DengLWF15}


\section{Basic construction}\label{sec:basic}

Our basic CoveringLSH construction is a Hamming projection family of the form (\ref{eq:hamming-projection}).
We start by observing the following simple property of Hamming projection families:

\begin{lemma}\label{lem:xor}
For every $\mathcal{A}\subseteq \{0,1\}^d$, every $h\in \mathcal{H}_\mathcal{A}$, and all $x,y\in \{0,1\}^d$ we have $h(x)=h(y)$ if and only if $h(x\oplus y) = {\bf 0}$.
\end{lemma}
\begin{proof}
Let $a\in \mathcal{A}$ be the vector such that $h(x) = x \wedge a$.
We have $h(x)=h(y)$ if and only if $a_i \ne 0 \Rightarrow x_i = y_i$. 
Since $x_i = y_i \Leftrightarrow (x \oplus y)_i = 0$ the claim follows. 
\end{proof}

Thus, to make sure all pairs of vectors within distance $r$ collide for some function, we need our family to have the property (implicit in the work of~\cite{Arasu_VLDB06}) that every vector with 1s in $r$ bit positions is mapped to zero by some function, i.e., the set of 1s is ``covered'' by zeros in a vector from $\mathcal{A}$.

\begin{definition}\label{def:covering}
For $\mathcal{A} \subseteq \{0,1\}^d$, the Hamming projection family $\mathcal{H}_\mathcal{A}$ is \emph{$r$-covering} if for every $x \in \{0,1\}^d$ with $||x|| \leq r$, there exists $h\in \mathcal{H}_\mathcal{A}$ such that $h(x) = {\bf 0}$.
The family is said to have \emph{weight} $\omega$ if $||a||\geq \omega d$ for every $a\in\mathcal{A}$.
\end{definition}

A trivial $r$-covering family uses $\mathcal{A} = \{{\bf 0}\}$.
We are interested in $r$-covering families that have a nonzero weight chosen to make collisions rare among vectors that are not close.
Vectors in our \emph{basic} $r$-covering family, which aims at weight around 1/2, will be indexed by nonzero vectors in $\{0,1\}^{r+1}$.
The family depends on a function $m: \{1,\dots,d\} \rightarrow \{0,1\}^{r+1}$ that maps bit positions to bit vectors of length~$r+1$. 
(We remark that if $d \leq 2^{r+1}-1$ and $m$ is the function that maps an integer to its binary representation, our construction is identical to known coverings based on finite geometry~\cite{gordon1995new}; however we give an elementary presentation that does not require knowledge of finite geometry.)
Define a family of bit vectors $a(v) \in \{0,1\}^d$ by
\begin{equation}\label{def:A_v}
 \tilde{a}(v)_i = \left\{ \begin{array}{ll}
 0 & \text{if } \langle m(i),v\rangle \equiv 0 \text{ mod } 2,\\
 1 & \text{ otherwise}
 \end{array} \right. \enspace .
\end{equation}
where $\langle m(i), v\rangle$ is the dot product of vectors $ m(i)$ and $v$.
We will consider the family of all such vectors with nonzero $v$:
\begin{equation*}
\mathcal{A}(m) = \big\{ a(v) \; | \; v\in \{0,1\}^{r+1} \backslash \{{\bf 0}\} \big\} \enspace .
\end{equation*}
Figure~\ref{fig:hadamard7} shows the family $\mathcal{A}(m)$ for $r=2$ and $m(i)$ equal to the binary representation of~$i$.

\begin{figure}[t]
	\begin{center}
\includegraphics{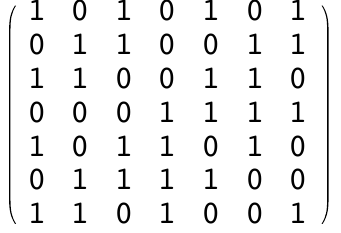}
\caption{The collection $\mathcal{A}_7$ corresponding to nonzero vectors of the Hadamard code of message length~3. The resulting Hamming projection family  $\mathcal{H}_{\mathcal{A}_7}$, see (\ref{eq:hamming-projection}), is $2$-covering since for every pair of columns there exists a row with 0s in these columns. It has weight $4/7$ since there are four 1s in each row.
Every row covers 3 of the 21 pairs of columns, so no smaller $2$-covering family of weight $4/7$ exists.}\label{fig:hadamard7}
	\end{center}
\end{figure}

\begin{lemma}\label{lemma:covering1}
For every $m: \{1,\dots,d\} \rightarrow \{0,1\}^{r+1}$, the Hamming projection family $\mathcal{H}_{\mathcal{A}(m)}$ is $r$-covering.
\end{lemma}
\begin{proof}
Let $x\in \{0,1\}^d$ satisfy $||x|| \leq r$ and consider $a(v)\in \mathcal{A}(m)$ as defined in (\ref{def:A_v}).
It is clear that whenever $i\in \{1,\dots,d\} \backslash I(x)$ we have $(a(v) \wedge x)_i = 0$ (recall that $I(x) = \{ i \; | \; x_i = 1\}$).
To consider $(a(v)\wedge x)_i$ for $i\in I(x)$ let $M_x = m(I(x))$, where elements are interpreted as $r+1$-dimensional vectors over the field ${\bf F}_2$.
The span of $M_x$ has dimension at most $|M_x| \leq ||x|| \leq r$, and since the space is $r+1$-dimensional there exists a vector $v_x\ne {\bf 0}$ that is orthogonal to $\text{span}(M_x)$.
In particular $\langle v_x,m(i)\rangle \text{ mod } 2 = 0$ for all $i\in I(x)$.
In turn, this means that $a(v_x) \wedge x = {\bf 0}$, as desired. 
\end{proof}

If the values of the function $m$ are ``balanced'' over nonzero vectors the family $\mathcal{H}_{\mathcal{A}(m)}$ has weight close to~$1/2$ for $d\gg 2^r$. 
More precisely we have:
\begin{lemma}\label{lemma:weight}
Suppose $|m^{-1}(v)|\geq \floor{d/2^{r+1}}$ for each $v\in \{0,1\}^{r+1}$ and $m^{-1}({\bf 0}) = \emptyset$. Then $\mathcal{H}_{\mathcal{A}(m)}$ has weight at least $2^r \floor{d/2^{r+1}}/d > \left(1-\tfrac{2^r}{d}\right)/2$.
\end{lemma}
\begin{proof}
It must be shown that $||a(v)|| \geq 2^r \floor{d/2^{r+1}}$ for each nonzero vector $v$.
Note that $v$ has a dot product of~1 with a set $V \subseteq \{0,1\}^{r+1}$ of exactly $2^{r}$ vectors (namely the nontrivial coset of $v$'s orthogonal complement).
For each $v'\in V$ the we have $a(v)_i = 1$ for all $i\in m^{-1}(v')$.
Thus the number of 1s in $a(v)$ is:
$$ \sum_{v'\in V} |m^{-1}(v')| \geq 2^r \floor{d/2^{r+1}} > \left(1-\tfrac{2^r}{d}\right)d/2\enspace . $$  
\end{proof}

\medskip

\paragraph{Comment on optimality}
We note that the size $|\mathcal{H}_{\mathcal{A}(m)}| = 2^{r+1}-1$ is close to the smallest possible for an $r$-covering families with weight around $1/2$.
To see this, observe that $\binom{d}{r}$ possible sets of errors need to be covered, and each hash function can cover at most $\binom{d/2}{r}$ such sets.
This means that the number of hash functions needed is at least 
$$\frac{\binom{d}{r}}{\binom{d/2}{r}} > 2^r$$
which is within a factor of~2 from the upper bound. $\triangle$

\medskip

Lemmas~\ref{lemma:covering1} and~\ref{lemma:weight} leave open the choice of mapping~$m$.
We will analyze the setting where $m$ maps to values chosen uniformly and independently from $\{0,1\}^{r+1}$. 
In this setting the condition of Lemma~\ref{lemma:weight} will in general not be satisfied, but it turns out that it suffices for $m$ to have balance in an expected sense.
We can relate collision probabilities to Hamming distances as follows:

\begin{theorem}\label{thm:basic}
	For all $x,y\in \{0,1\}^d$ and for random $m: \{1,\dots,d\} \rightarrow \{0,1\}^{r+1}$,
	\begin{enumerate}
\item If $||x-y||\leq r$ then $\Pr{\exists h\in \mathcal{H}_{\mathcal{A}(m)}: h(x)=h(y)} = 1$.
\item $\E{ | \{ h\in \mathcal{H}_{\mathcal{A}(m)} \;|\; h(x)=h(y) \} | } < 2^{r+1-||x-y||}$.
\end{enumerate}
\end{theorem}
\begin{proof}
Let $z = x \oplus y$.
For the first part we have $||x-y|| = ||z|| \leq r$.
Lemma~\ref{lemma:covering1} states that there exists 
$h\in \mathcal{H}_{\mathcal{A}(m)}$ such that $h(z)= {\bf 0}$.
By Lemma~\ref{lem:xor} this implies $h(x)=h(y)$.

To show the second part we fix $v\in \{0,1\}^{r+1} \backslash \{{\bf 0}\}$.
Now consider $a(v)\in\mathcal{A}(m)$, defined in~(\ref{def:A_v}), and the corresponding function $h(x) = x \wedge a(v) \in \mathcal{H}_{\mathcal{A}(m)}$.
For $i\in I(z)$ we have $h(z)_i = 0$ if and only if $a(v)_i = 0$.
Since $m$ is random and $v\ne {\bf 0}$ the $a(v)_i$ values are independent and random, so the probability that $a(v)_i = 0$ for all $i\in I(z)$ is $2^{-||z||} = 2^{-||x-y||}$.
By linearity of expectation, summing over $2^{r+1}-1$ choices of $v$ the claim follows. 
\end{proof}

\paragraph{Comments} A few remarks on Theorem~\ref{thm:basic} (that can be skipped if the reader wishes to proceed to the algorithmic results):
\begin{itemize}
\item The vectors in $\mathcal{A}(m)$ can be seen as samples from a \emph{Hadamard code} consisting of $2^{r+1}$ vectors of dimension $2^{r+1}$, where bit $i$ of vector $j$ is defined by $\langle i,j\rangle \text{ mod }2$, again interpreting the integers $i$ and $j$ as vectors in ${\bf F}_2^d$.
Nonzero Hadamard codewords have Hamming weight and minimum distance $2^{r+1}$.
However, it does not seem that error-correcting ability in general yields nontrivial $r$-covering families.

\item The construction can be improved by changing $m$ to map to $\{0,1\}^{r+1}\backslash \{{\bf 0}\}$ and/or requiring the function values of $m$ to be \emph{balanced} such that the number of bit positions mapping to each vector in $\{0,1\}^{r+1}$ is roughly the same.
This gives an improvement when $d\approx 2^r$ but is not significant when $d$ is much smaller or much larger than $2^r$.
To keep the exposition simple we do not analyze this variant.
\item At first glance it appears that the ability to avoid collision for CoveringLSH (``filtering'') is not significant when $||x-y|| = r+1$.
However, we observe that for similarity search in Hamming space it can be assumed without loss of generality that either all distances from the query point are even or all distances are odd.
This can be achieved by splitting the data set into two parts, having even and odd Hamming weight, respectively, and handling them separately.
For a given query $y$ and radius~$r$ we then perform a search in each part, one with radius~$r$ and one with radius $r-1$ (in the part of data where distance $r$ to $y$ is not possible).
This reduces the expected number of collisions at distance $r+1$ to at most $1/2$. $\triangle$
\end{itemize}

\paragraph{Nearest neighbor} 

Above we have assumed that the search radius $r$ was given in advance, but it turns out that CoveringLSH supports also supports finding the nearest neighbor, under the condition that the distance is at most~$r$.
To see this, consider the subfamily of $\mathcal{A}(m)$ indexed by vectors of the form $0^{r+1-r_1} v_1$, where $v_1\in \{0,1\}^{r_1+1}\backslash\{{\bf 0}\}$ for some $r_1\leq r$, then collision is guaranteed up to distance $r_1$.
That is, we can search for a nearest neighbor at an unknown distance in a natural way, by letting $m$ map randomly to $\{0,1\}^{\lceil \log n\rceil}$ and choosing $v$ as the binary representation of $1,2,3,\dots$ (or alternatively, the vectors in a Gray code for $\{0,1\}^{\lceil \log n\rceil}$).
In either case Theorem~\ref{thm:basic} implies the invariant that the nearest neighbor has distance at least $\lfloor \log v \rfloor$, where $v$ is interpreted as an integer.
This means that when a point $x$ at distance at most $c\, \lfloor \log (v+1) \rfloor$ is found, we can stop after finishing iteration $v$ and return $x$ as a $c$-approximate nearest neighbor.
Figure~\ref{fig:code} gives pseudocode for data structure construction and nearest neighbor queries using CoveringLSH.\footnote{A corresponding Python implementation is available on github, {\tt https://github.com/rasmus-pagh/coveringLSH}.}

\begin{figure}[t]
	\hrule
	\begin{minipage}{.5\linewidth}
	\begin{tabbing}
	  xx\=xx\=xx\=xx\=xx\=xx\=xx\=\kill\\
	  {\bf procedure} {\sc InitializeCovering}$(d,r)$\+\\
	    {\bf for} $v\in \{0,1\}^{r+1}$ {\bf do} $A[v] := 0^{d}$\\
	    {\bf for} $i:=1$ to $d$ {\bf do} \+\\
			$ m := \text{\sc Random}(\{0,1\}^{r+1}\backslash \{{\bf 0}\})$\\
	    	{\bf for} $v\in \{0,1\}^{r+1}$ {\bf do}
				$A[v]_i := \langle m, v\rangle$ mod $2$\-\\
	  {\bf end for}\-\\
	  {\bf end}\\
	  \\
	  {\bf function} {\sc BuildDataStructure}$(S,r)$\+\\
	  	$D = \emptyset$\\
	    {\bf for} $x\in S$, $v\in \{0,1\}^{r+1}\backslash \{{\bf 0}\}$ {\bf do}\+\\
		      $D[x \wedge A[v]] := D[x \wedge A[v]] \cup \{x\}$\-\\
		{\bf return} $D$\-\\
	  {\bf end}\\
	\end{tabbing}
	\end{minipage}
	\begin{minipage}{.5\linewidth}
	\begin{tabbing}
	  xx\=xx\=xx\=xx\=xx\=xx\=xx\=\kill
	  \+\+\\
	  {\bf function} {\sc NearestNeighbor}$(D,r,y)$\+\\
	  	$best := \infty$\\ 
	  	$nn := {\tt null}$\\ 
	    {\bf for} $v:=1$ to $2^{r+1}-1$ {\bf do} \+\\
			{\bf for} $x \in D[y \wedge A[\text{\sc BitVec}(v,r+1)]]$ {\bf do}\+\\
		  		{\bf if} $||x-y|| < best$ {\bf then}\+\\
					$best = ||x-y||$\\
					$nn = x$\-\\
				{\bf end if}\-\\
			{\bf end for}\\
			{\bf if} $best \leq \lfloor \log(v+1)\rfloor$ {\bf then return} $nn$\-\\
		{\bf end for}\\
		{\bf return} ${\tt null}$\-\\
	  {\bf end}
	\end{tabbing}
	\vspace{2.5mm}
	\end{minipage}
	\hrule
\caption{Pseudocode for constructing (left) and querying (right) a nearest neighbor data structure on a set $S\subseteq \{0,1\}^d$ as described in section~\ref{sec:randomdata}.
Parameter $r$ controls the largest radius for which a nearest neighbor is returned.
This is the simplest instantiation of CoveringLSH --- it works well on high-entropy data where there are few points within distance $r + \log_2 |S|$ of a query point.
In this setting, given a query point $y$, the expected search time for finding a nearest neighbor $x$ is $\BO{2^{||x-y||}}$. If only a $c$-approxiate nearest neighbor is sought the condition $best \leq \lfloor \log(v+1)\rfloor$ should be changed to $best \leq c \lfloor \log(v+1)\rfloor$. \smallskip\newline
{\bf Notation:} The function $\text{\sc Random}$ returns a random element from a given set.
The inner product $\langle m, v\rangle$ can be computed by a bitwise conjunction followed by counting the number of bits set ({\sc popcnt}).
$D[i]$ is used to denote the information associated with key $i$ in the dictionary $D$ that is the main part of the data structure; if $i$ is not a key in $D$ then $D[i]=\emptyset$. 
The function call $\text{\sc BitVec}(v,r+1)$ typecasts an integer to a bit vector of dimension $r+1$.
Finally, $||x-y||$ denotes the Hamming distance between $x$ and $y$. 
\smallskip\newline
{\bf Other comments:} 
Vectors are stored $2^{r+1}-1$ times in $D$, but may be represented as references to a single occurrence in memory to achieve better space complexity for large~$d$.
The global dictionary $A$, which contains a covering independent of the set $S$, must be initialized by {\sc InitializeCovering} before {\sc BuildDataStructure} is called.
Note that the function $m$ is not stored, as it is not needed after constructing the covering.}
\label{fig:code}
\end{figure}

\subsection{Approximation factor $c=\log(n)/r$}\label{sec:randomdata}

We first consider a case in which the method above directly gives a strong result, namely when the threshold $cr$ for being an approximate near neighbor equals $\log n$.
Such a threshold may be appropriate for high-entropy data sets of dimension $d>2\log n$ where most distances tend to be large (see~\cite{kucherov2005multiseed,norouzi2012fast} for discussion of such settings).
In this case Theorem~\ref{thm:basic} implies efficient $c$-approximate near neighbor search in expected time $\BO{\Delta_S 2^r} = \BO{\Delta_S\, n^{1/c}}$, where $\Delta_S$ bounds the time to compute the Hamming distance between query vector $y$ and a vector $x\in S$. This matches the asymptotic time complexity of~\cite{Indyk1998}.

To show this bound observe that the expected total number of collisions $h(x)=h(y)$, summed over all $h\in \mathcal{H}_{\mathcal{A}(m)}$ and $x\in S$ with $||x-y||\geq \log n$, is at most $2^{r+1}$.
This means that computing $h(y)$ for each $h\in \mathcal{H}_{\mathcal{A}(m)}$ and computing the distance to the vectors that are not within distance $cr$ but collide with $y$ under some $h\in \mathcal{H}_{\mathcal{A}(m)}$ can be done in expected time $\BO{\Delta_S 2^r}$.
The expected bound can be supplemented by a high probability bound as follows: Restart the search in a new data structure if the expected time is exceeded by a factor of~2. Use $\BO{\log n}$ data structures and resort to brute force if this fails, which happens with polynomially small probability in~$n$.

What we have bounded is in fact performance on a \emph{worst case} data set in which most data points are just above the threshold for being a $c$-approximate near neighbor.
In general the amount of time needed for a search will depend on the distribution of distances between~$y$ and data points, and may be significantly lower.

The space required is $\BO{2^r n} = \BO{n^{1+1/c}}$ words plus the space required to store the vectors in~$S$, again matching the bound of Indyk and Motwani.
In a straightforward implementation we need additional space $\BO{d}$ to store the function $m$, but if $d$ is large (for sets of sparse vectors) we may reduce this by only storing $m(i)$ if there exists $x\in S$ with $x_i \ne 0$.
With this modification, storing $m$ does not change the asymptotic space usage.
For dense vectors it may be more desirable to explicitly store the set of covering vectors $\mathcal{A}(m)$ rather than the function $m$, and indeed this is the approach taken in the pseudocode.

\medskip

\paragraph{Example} Suppose we have a set $S$ of $n=2^{30}$ vectors from $\{0,1\}^{128}$ and wish to search for a vector at distance at most $r=10$ from a query vector $y$.
A brute-force search within radius~$r$ would take much more time than linear search, so we settle for $3$-approximate similarity search.
Vectors at distance larger than $3r$ have collision probability at most $1/(2n)$ under each of the $2^{r+1}-1$ functions in $h\in \mathcal{H}_{\mathcal{A}(m)}$, so in expectation there will be less than $2^r = 1024$ hash collisions between $y$ and vectors in $S$. 
The time to answer a query is bounded by the time to compute $2047$ hash values for $y$ and inspect the hash collisions.

It is instructive to compare to the family $\mathcal{H}_{\mathcal{A}(R)}$ of Indyk and Motwani, described in section~\ref{sec:lsh}, with the same performance parameters ($2047$ hash evaluations, collision probability $1/(2n)$ at distance $31$).
A simple computation shows that for $k=78$ samples we get the desired collision probability, and collision probability $(1-r/128)^{78} \approx 0.0018$ at distance~$r=10$.
This means that the probability of a false negative by not producing a hash collision for a point at distance $r$ is $(1-(1-r/128)^{78})^{2047} > 0.027$.
So the risk of a false negative is nontrivial given the same time and space requirements as our ``covering'' LSH scheme. $\triangle$

%


\section{Construction for large distances}\label{sec:large}

Our basic construction is only efficient when $cr$ has the ``right'' size (not too small, not too large). 
We now generalize the construction to arbitrary values of $r$, $cr$, and $n$, with a focus on efficiency for large distances. 
In a nutshell:
\begin{itemize}
\item For an arbitrary choice of $cr$ (even much larger than $\log n$) we can achieve performance that differs from classical LSH by a factor of $\ln(4) < 1.4$ in the exponent.
\item We can match the exponent of classical LSH for the $c$-approximate near neighbor problem whenever $\lceil\log n\rceil/(cr)$ is (close to) integer.
\end{itemize}
\noindent
We still use a Hamming projection family (\ref{eq:hamming-projection}), changing only the set $\mathcal{A}$ of bit masks used.
Our data structure will depend on parameters $c$ and $r$, i.e., these can not be specified as part of a query.
Without loss of generality we assume that $cr$ is integer.

\medskip

\paragraph{Intuition}
When $cr < \log n$ we need to increase the average number of 1s in the bit masks to reduce collision probabilities.
The increase should happen in a correlated fashion in order to maintain the guarantee of collision at distance~$r$.
The main idea is to increase the fraction of 1s from $1/2$ to $1-2^{-t}$, for $t\in{\bf N}$, by essentially repeating the sampling from the Hadamard code $t$ times and selecting those positions where at least one sample hits a~1. 

On the other hand, when $cr > \log n$ we need to decrease the average number of 1s in the bit masks to increase collision probabilities. This is done using a refinement of the partitioning method of~\cite{Arasu_VLDB06} which distributes the dimensions across partitions in a balanced way.
The reason this step does not introduce false negatives is that for each data point $x$ there will always exist a partition in which the distance between query $y$ and~$x$ is at most the average across partitions. 
An example is shown in figure~\ref{fig:hadamard2x7}.
$\triangle$
\medskip

\begin{figure}[t]
	\begin{center}
\includegraphics{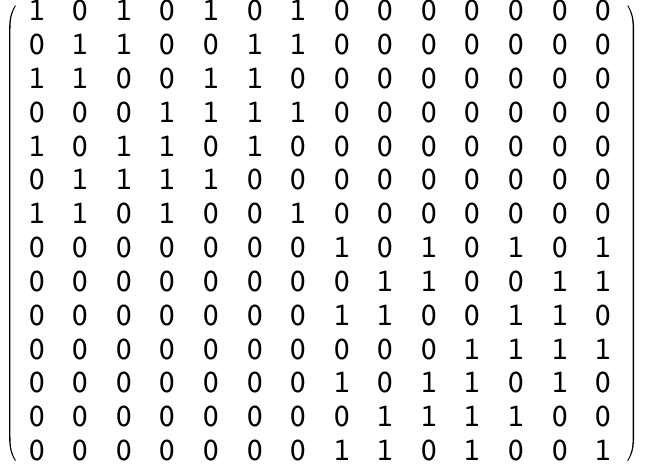}
\caption{The collection $\mathcal{A}_{2\times7}$ containing two copies of the collection $\mathcal{A}_7$ from figure~\ref{fig:hadamard7}, one for each half of the dimensions. The resulting Hamming projection family  $\mathcal{H}_{\mathcal{A}_{2\times7}}$, see (\ref{eq:hamming-projection}), is $5$-covering since for set of 5 columns there exists a row with 0s in these columns. It has weight $4/14$ since there are four 1s in each row.
Every row covers $\binom{10}{5}$ sets of 5 columns, so a lower bound on the size of a $5$-covering collection of weight $4/14$ is $\lceil \tbinom{14}{5}/\tbinom{10}{5}\rceil = 8$.}\label{fig:hadamard2x7}
	\end{center}
\end{figure}

We use $b,q\in {\bf N}$ to denote, respectively, the number of partitions and the number of partitions to which each dimension belongs.
Observe that if we distribute $q$ copies of $r$ ``mismatching'' dimensions across $b$ partitions, there will always exist a partition with at most $r' = \lfloor r q/b \rfloor$ mismatches.
Let $\text{Intervals}(b,q)$ denote the set of intervals in $\{1,\dots,b\}$ of length $q$, where intervals are considered modulo~$b$ (i.e., with wraparound).
We will use two random functions,
$$m: \{1,\dots,d\} \rightarrow \left(\{0,1\}^{t r' + 1} \right)^t$$
$$s: \{1,\dots,d\} \rightarrow \text{Intervals}(b,q)$$
to define a family of bit vectors $a(v,k) \in \{0,1\}^d$, indexed by vectors $v\in \{0,1\}^{tr'+1}$ and $k\in \{1,\dots,b\}$. 
We define a family of bit vectors $a(v,k)\in \{0,1\}^d$ by
\begin{equation}\label{def:A_vk}
a(v,k)_i = s^{-1}(k)_i \wedge \left(\bigvee_j \langle m(i)_j, v\rangle \text{ mod } 2 \ne 0\right),
\end{equation}
where $s^{-1}(k)$ is the preimage of $k$ under $s$ represented as a vector in $\{0,1\}^d$ (that is, $s^{-1}(k)_i = 1$ if and only if $s(i)=k$), and $\langle m(i)_j, v\rangle$ is the dot product of vectors $ m(i)_j$ and~$v$.
We will consider the family of all such vectors with nonzero $v$:
$$\mathcal{A}(m,s) = \big\{ a(v,k) \; | \; v\in \{0,1\}^{t r'+1}\backslash\{{\bf 0}\},\; k\in\{1,\dots,b\} \big\} \enspace .$$
Note that the size of $\mathcal{A}(m,s)$ is $b\, (2^{t r'+1}-1) < 2b\, 2^{trq/b}$.

\begin{lemma}\label{lemma:covering2}
For every choice of $b,d,q,t\in {\bf N}$, and every choice of functions $m$ and $s$ as defined above, the Hamming projection family $\mathcal{H}_{\mathcal{A}(m,s)}$ is $r$-covering.
\end{lemma}
\begin{proof}
Let $x\in\{0,1\}^d$ satisfy $||x||\leq r$. 
We must argue that there exists a vector $v^*\in \{0,1\}^{tr'+1}\backslash\{{\bf 0}\}$ and $k^*\in \{1,\dots,b\}$ such that $a(v^*,k^*)\wedge x = {\bf 0}$, i.e., by (\ref{def:A_vk})
$$\forall i: x_i \wedge s^{-1}(k)_i \wedge (\bigvee_j \langle m(i)_j, v\rangle \text{ mod } 2 \ne 0) = 0 \enspace .$$
We let $k^* = \arg\min || x \wedge s^{-1}(k) ||$, breaking ties arbitrarily.
Informally, $k^*$ is the partition with the smallest number of 1s in $x$.
Note that $\sum_{k=1}^b || x \wedge s^{-1}(k) || = qr$ so by the pigeonhole principle, $|| x \wedge s^{-1}(k^*) || \leq \lfloor r q/b \rfloor = r'$.
Now consider the ``problematic'' set $I(x \wedge s^{-1}(k^*))$ of positions of 1s in $x \wedge s^{-1}(k^*)$, and the set of vectors that $m$ associates with it:
$$M_x = \{ m(I(x \wedge s^{-1}(k^*)))_j \; | \; j\in \{1,\dots,t\}\} \enspace .$$
The span of $M_x$ has dimension at most $|M_x| \leq tr'$.
This means that there must exist $v^*\in \{0,1\}^{tr'+1}\backslash\{{\bf 0}\}$ that is orthogonal to all vectors in $M_x$.
In particular this implies that for each $i\in I_x$ we have $\bigvee_j \langle m(i)_j,v^*\rangle \text{ mod 2} \ne 0$ is false, as desired. 
\end{proof}

We are now ready to show the following extension of Theorem~\ref{thm:basic}:
\begin{theorem}\label{thm:main}
	For random $m$ and $s$, for every $b,d,q,r,t\in {\bf N}$ and $x,y\in \{0,1\}^d$:
	\begin{enumerate}
\item $||x-y|| \leq r \Rightarrow \Pr{\exists h\in \mathcal{H}_{\mathcal{A}(m,s)}: h(x)=h(y)} = 1$.
\item $\E{ | \{ h\in \mathcal{H}_{\mathcal{A}(m,s)} \;|\; h(x)=h(y) \} | }
	< \left(1 - (1-2^{-t}) q/b \right) ^ {||x-y||} b\, 2^{trq/b+1}$.
\end{enumerate}
\end{theorem}
\begin{proof}
    By Lemma~\ref{lem:xor} we have $h(x)=h(y)$ if and only if $h(z)={\bf 0}$ where $z = x \oplus y$.
	So the first part of the theorem is a consequence of Lemma~\ref{lemma:covering2}.
	For the second part consider a particular vector $a(v,k)$, where $v$ is nonzero, and the corresponding hash value $h(z) = z \wedge a(v,k)$.
	We argue that over the random choice of $m$ and $s$ we have, for each $i$:
	\begin{align}\label{eq:zeroprob}
	\Pr{ a(v,k)_i = 0 }  \;=\; & \Pr{ s^{-1}(k)_i = 0 } + \Pr{s^{-1}(k)_i = 1  \; \wedge \; \forall j: \langle m(i)_j,v\rangle \equiv 0 \text{ mod } 2} \nonumber \\
	\;=\; & (1-q/b) + 2^{-t} q/b \\
	\;=\; & 1 - (1-2^{-t})q/b \nonumber \enspace .
	\end{align}
	The second equality uses independence of the vectors $\{ m(i)_j \; | \; j=1,\dots,t\}$ and $s(i)$, and that for each $j$ we have $\Pr{\langle m(i)_j,v\rangle \equiv 0 \text{ mod }2} = 1/2$.
	Observe also that $a(v,k)_i$ depends only on $s(i)$ and $m(i)$.
	Since function values of~$s$ and $m$ are independent, so are the values 
	$$\{ a(v,k)_i \; | \; i\in\{1,\dots,d\}\} \enspace .$$
	This means that the probability of having $a(v,k)_i = 0$ for all $i$ where $z_i=1$ is a product of probabilities from~(\ref{eq:zeroprob}):
\begin{align*}
	\Pr{h(x)=h(y)} & = \prod_{i\in I_z} \left(1 - (1-2^{-t})q/b\right)
 	=\left(1 - (1-2^{-t}) q/b \right)^{||x-y||} \enspace .
 \end{align*}
The second part of the theorem follows by linearity of expectation, summing over the 
vectors in $\mathcal{A}(m,s)$. 
\end{proof}

\subsection{Choice of parameters}

The expected time complexity of $c$-approximate near neighbor search with radius~$r$ is bounded by the size $|\mathcal{A}|$ of the hash family plus the expected number $\kappa_\mathcal{A}$ of hash collisions between the query~$y$ and vectors $S$ that are not $c$-approximate near neighbors. Define
$$S_\text{far} = \{ x\in S \; | \; ||x-y|| > cr \} \text{ and }
\kappa_\mathcal{A} = \E{ | \{ (x,h)\in S_\text{far} \times \mathcal{H}_{\mathcal{A}} \;|\; h(x)=h(y) \} | }$$
where the expectation is over the choice of family $\mathcal{A}$.
Choosing parameters $t$, $b$, and $q$ in Theorem~\ref{thm:main} in order to get a family $\mathcal{A}$ that minimizes $|\mathcal{A}| + \kappa_\mathcal{A}$ is nontrivial.
Ideally we would like to balance the two costs, but integrality of the parameters means that there are ``jumps'' in the possible sizes and filtering efficiencies of
$\mathcal{H}_{\mathcal{A}(m,s)}$. Figure~\ref{fig:numerical-optimization-low-radius} shows bounds achieved by numerically selecting the best parameters in different settings. We give a theoretical analysis of some choices of interest below. In the most general case the strategy is to reduce to a set of subproblems that hit the ``sweet spot'' of the method, i.e., where $|\mathcal{A}|$ and $\kappa_\mathcal{A}$ can be made equal. $\triangle$

\begin{figure*}[t]
	\begin{center}
	\includegraphics[width=\textwidth]{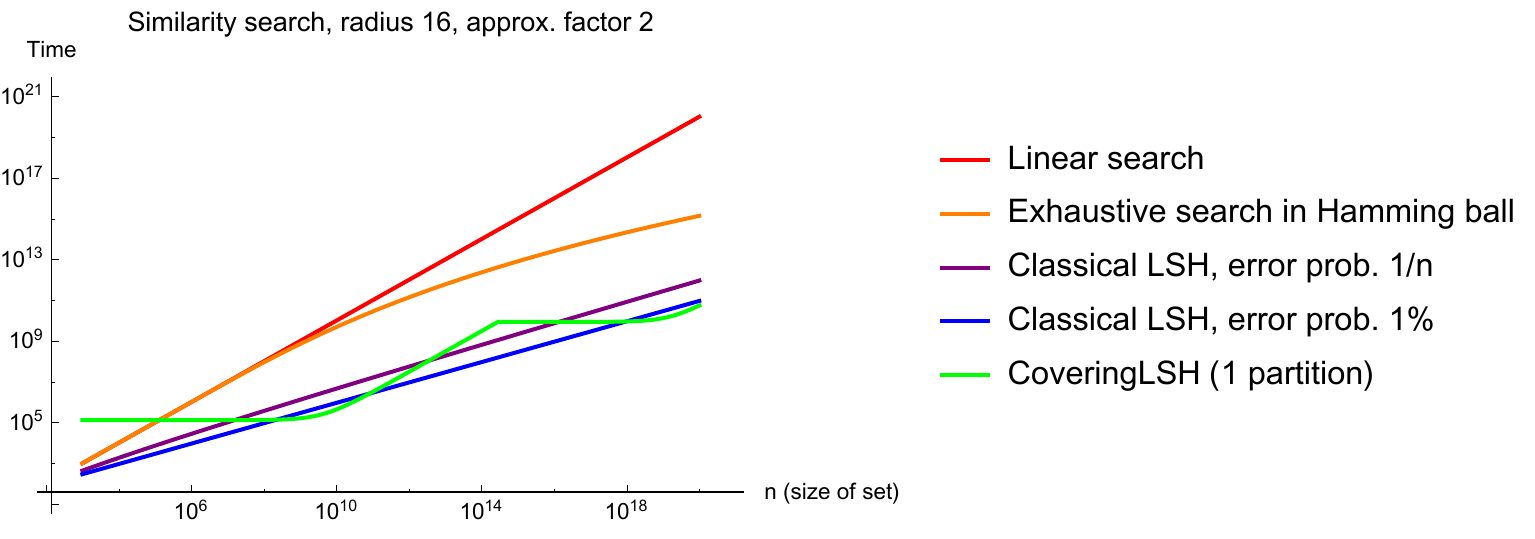}\\
	\bigskip
	\includegraphics[width=\textwidth]{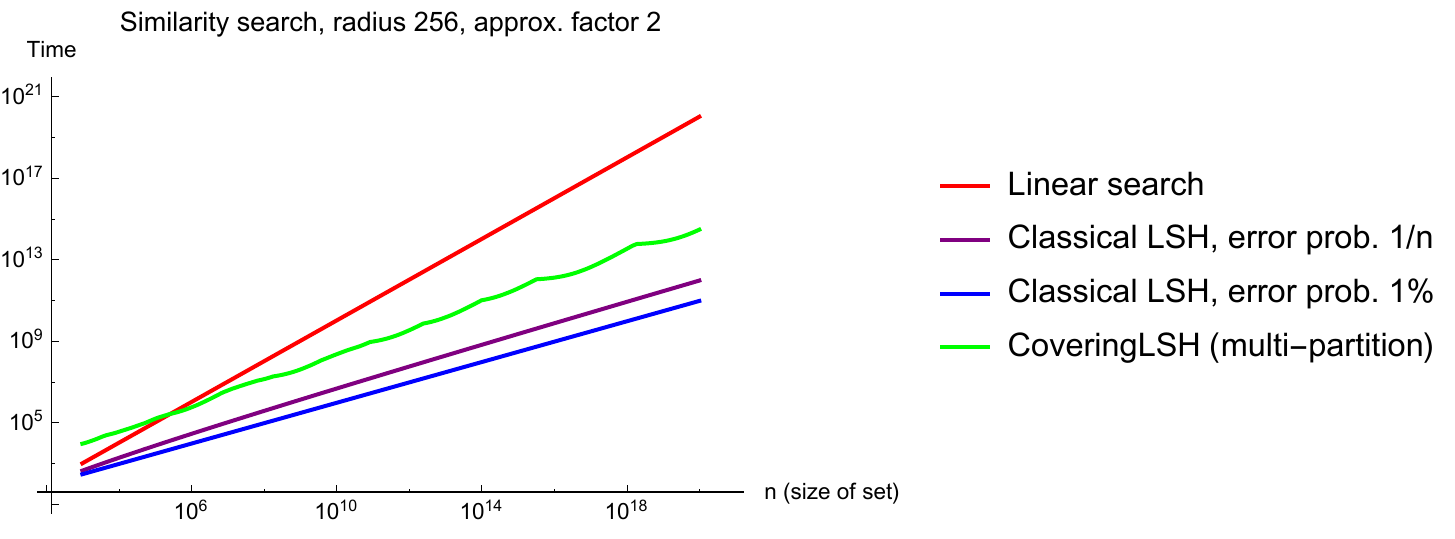}
	\caption{Expected number of memory accesses for different similarity search methods for finding a vector within Hamming distance~$r$ of a query vector~$y$. The plots are for $r=16$ and $r=256$, respectively, and are for a worst-case data set where all points have distance~$2r$ from~$y$, i.e., there exists no $c$-approximate near neighbor for an approximation factor $c<2$. The bound for exhaustive search in a Hamming ball of radius $r$ optimistically assumes that the number of dimensions is~$\log_2 n$, which is smallest possible for a data set of size $n$ (for $r=256$ this number is so large that it is not even shown).
	Two bounds are shown for the classical LSH method of Indyk and Motwani: A small fixed false negative probability of 1\%, and a false negative probability of $1/n$.
	The latter is what is needed to ensure no false negatives in a sequence of $n$ searches.
	The bound for CoveringLSH in the case $r=16$ uses a single partition ($b=1$), while for $r=256$ multiple partitions are used.}\label{fig:numerical-optimization-low-radius}
	\end{center}
\end{figure*}

\begin{corollary}\label{cor:params}
	For every $c>1$ there exist explicit, randomized $r$-covering Hamming projection families $\mathcal{H}_{\mathcal{A}_1}$, $\mathcal{H}_{\mathcal{A}_2}$
	such that for every $y\in \{0,1\}^d$:
	\begin{enumerate}
		\item $|\mathcal{A}_1| \leq 2^{r+1} n^{1/c}$ and 
		$\kappa_{\mathcal{A}_1} < 2^{r+1} n^{1/c}$.
		\item If $\log(n)/(cr) + \varepsilon \in {\bf N}$, for $\varepsilon > 0$, then $|\mathcal{A}_1| \leq 2^{\varepsilon r+1} n^{1/c}$ and $\kappa_{\mathcal{A}_1} < 2^{\varepsilon r+1} n^{1/c}$.
		\item If $r>\ceil{\ln(n)/c}$ then $|\mathcal{A}_2| \leq 8r\, n^{\ln(4)/c}$ and 
		$\kappa_{\mathcal{A}_2} < 8r\, n^{\ln(4)/c}$.
	\end{enumerate}
\end{corollary}
\begin{proof}
		We let $\mathcal{A}_1 = \mathcal{A}(m,s)$ with $b=q=1$ and $t = \ceil{\log(n)/(cr)}$. Then 
		$$|\mathcal{A}_1| < 2b\, 2^{trq/b} = b\, 2^{tr+1} \leq 2^{(\log(n)/(cr)+1)r+1} = 2^{r+1} n^{1/c} \enspace .$$
		Summing over $x\in S_\text{far}$ the second part of Theorem~\ref{thm:main} yields: 
		$$\kappa_{\mathcal{A}_1} < n 2^{-tcr} 2^{tr+1} \leq 2^{tr+1} \leq 2^{r+1} n^{1/c} \enspace .$$

\medskip
		
		For the second bound on $\mathcal{A}_1$ we notice that the factor $2^r$ is caused by the rounding in the definition of~$t$, which can cause $2^{tr}$ to jump by a factor $2^r$. When $\log(n)/(cr) + \varepsilon$ is integer we instead get a factor $2^{\varepsilon r}$.
		
\medskip
		
		Finally, we let $\mathcal{A}_2 = \mathcal{A}(m,s)$ with $b=r$, $q=2\ceil{\ln(n)/c}$, and $t = 1$. The size of $\mathcal{A}_2$ is bounded by 
		$b\, 2^{trq/b+1} \leq r\, 2^{2\ln(n)/c + 3} = 8r\, n^{\ln(4)/c}$.
		Again, by Theorem~\ref{thm:main} and summing over $x\in S_\text{far}$:
		\begin{align*}
			\kappa_{\mathcal{A}_2} & < n \left(1 - q/(2r) \right)^{cr} r\, 2^{q+1}
			 < n \exp\left(- qc/2) \right) r\, 2^{q+1}
			 < r\, 2^{q+1}
			 < 8r\, n^{\ln(4)/c},
		\end{align*}
		where the second inequality follows from the fact that $1-\alpha < \exp(-\alpha)$ when $\alpha > 0$. 
\end{proof}


\section{Construction for small distances}\label{sec:small-radius}

In this section we present a different generalization of the basic construction of Section~\ref{sec:basic} that is more efficient for small distances, $cr\leq \log(n)/(3\log\log n)$, than the construction of Section~\ref{sec:large}.
The existence of \emph{asymptotically} good near neighbor data structures for small distances is not a big surprise: For $r = o(\log(n)/\log\log n)$ it is known how to achieve query time $n^{o(1)}$~\cite{Cole:2004:DMI:1007352.1007374}, even with $c=1$.
In practice this will most likely no faster than linear search for realistic values of~$n$ except when $r$ is a small constant.
In contrast we seek a method that has reasonable constant factors and may be useful in practice.

The idea behind the generalization is to consider vectors and dot products modulo~$p$ for some prime $p>2$.
This corresponds to using finite geometry coverings over the field of size $p$~\cite{gordon1995new}, but like in Section~\ref{sec:basic} we make an elementary presentation without explicitly referring to finite geometry.
Vectors in the $r$-covering family, which aims at weight around $1-1/p$, will be indexed by nonzero vectors in $\{0,\dots,p-1\}^{r+1}$.
Generalizing the setting of Section~\ref{sec:basic}, the family depends on a function $m: \{1,\dots,d\} \rightarrow \{0,\dots,p-1\}^{r+1}$ that maps bit positions to vectors of length~$r+1$. 
Define a family of bit vectors $\tilde{a}(v) \in \{0,1\}^d$, $v\in \{0,\dots,p-1\}^{r+1}$ by
\begin{equation}\label{def:Atilde_v}
 \tilde{a}(v)_i = \left\{ \begin{array}{ll}
 0 & \text{if } \langle m(i),v\rangle \equiv 0 \text{ mod } p,\\
 1 & \text{ otherwise}
 \end{array} \right. \enspace .
\end{equation}
for all $i\in\{1,\dots,d\}$, where $\langle m(i), v\rangle$ is the dot product of vectors $m(i)$ and $v$.
We will consider the family of all such vectors with nonzero $v$:
\begin{equation*}
\tilde{\mathcal{A}}(m) = \big\{ \tilde{a}(v) \; | \; v\in \{0,\dots,p-1\}^{r+1} \backslash \{{\bf 0}\} \big\} \enspace .
\end{equation*}


\begin{lemma}\label{lemma:covering3}
For every $m: \{1,\dots,d\} \rightarrow \{0,\dots,p-1\}^{r+1}$, the Hamming projection family $\mathcal{H}_{\tilde{\mathcal{A}}(m)}$ is $r$-covering.
\end{lemma}
\begin{proof}
	Identical to the proof of Lemma~\ref{lemma:covering1}. The only difference is that we consider the field ${\bf F}_p$ of size~$p$.
\end{proof}


Next, we relate collision probabilities to Hamming distances as follows:
\begin{theorem}\label{thm:finitegeometry}
	For all $x,y\in \{0,1\}^d$ and for random $m: \{1,\dots,d\} \rightarrow \{0,\dots,p-1\}^{r+1}$,
	\begin{enumerate}
\item If $||x-y||\leq r$ then $\Pr{\exists h\in \mathcal{H}_{\tilde{\mathcal{A}}(m)}: h(x)=h(y)} = 1$.
\item $\E{ | \{ h\in \mathcal{H}_{\tilde{\mathcal{A}}(m)} \;|\; h(x)=h(y) \} | } < p^{r+1-||x-y||}$.
\end{enumerate}
\end{theorem}
\begin{proof}
The proof is completely analogous to that of Theorem~\ref{thm:basic}. 
The first part follows from Lemma~\ref{lemma:covering3}.
For the second part we use that $\Pr{\langle m(i),v\rangle \equiv 0 \text{ mod } p} = 1/p$ for each $v\ne{\bf 0}$ and that we are summing over $p^{r+1}-1$ values of $v$.
\end{proof}

Now suppose that $cr \leq \log(n)/(3 \log\log n)$ and let $p$ be the smallest prime number such that $p^{cr} > n$, or in other words the smallest prime $p > n^{1/(cr)}$. We refer to the family $\tilde{\mathcal{A}}(m)$ with this choice of $p$ as $\mathcal{A}_3$, and note that $|\mathcal{A}_3| < p^{r+1}$.

By the second part of Theorem~\ref{thm:finitegeometry} the expected total number of collisions $h(x)=h(y)$, summed over all $h\in \mathcal{H}_{\mathcal{A}_3}$ and $x\in S$ with $||x-y||\geq cr$, is at most $p^{r+1}$.
This means that computing $h(y)$ for each $h\in \mathcal{H}_{\mathcal{A}_3}$ and computing the distance to the vectors that are not within distance $cr$ but collide with $y$ under some $h\in \mathcal{H}_{\mathcal{A}_3}$ can be done in expected time $\BO{\Delta_S\, p^r}$.

What remains is to bound $p^{r+1}$ in terms of $n$ and~$c$.
According to results on prime gaps (see e.g.~\cite{dudek2014explicit} and its references)
there exists a prime between every pair of cubes $\alpha^3$ and $(\alpha+1)^3$ for $\alpha$ larger than an explicit constant.
We will use the slightly weaker upper bound $(1+4/\alpha) \alpha^3 > (\alpha+1)^3$, which holds for $\alpha > 4$.
If $n$ exceeds a certain constant, since $p$ is the smallest such prime, choosing $\alpha = n^{1/(3cr)}$ we have $p < (1+4/\alpha) \alpha^3$.
By our upper bound on $cr$ we have $\alpha > n^{\log\log(n)/\log n} = \log n$.
Using $r+1 \leq \log(n)$ we have
\begin{equation}
|\mathcal{A}_3| < p^{r+1} < ((1+4/\alpha) \alpha^3)^{r+1} < (1+4/\log n)^{\log n} n^{\frac{r+1}{cr}} < e^4 n^{\frac{r+1}{cr}} \enspace .
\end{equation}

\paragraph{Improvement for small $r$}

To asymptotically improve this bound for small $r$ we observe that without loss of generality we can assume that $cr \geq \log(n)/(6 \log\log n)$: 
If this is not the case move to vectors of dimension $dt$ by repeating all vectors $t$ times, where $t$ is the largest integer with $crt \leq \log(n)/(3 \log\log n)$. 
This increases all distances by a factor exactly $t < \log n$, and increases $\Delta_S$ by at most a factor $t < \log n$.
Then we have:
\begin{equation}\label{eq:smallbound}
|\mathcal{A}_3| < p^{r+1} < n^{\frac{r+1}{cr}} = n^{1/c+1/(cr)} \leq n^{1/c} (\log n)^6 \enspace .
\end{equation}
That is, the expected time usage of $p^{r+1}$ matches the asymptotic time complexity of~\cite{Indyk1998} up to a polylogarithmic factor.

\paragraph{Comments}

In principle, we could combine the construction of this section with partitioning to achieve improved results for some parameter choices.
However, it appears difficult to use this for improved bounds in general, so we have chosen to not go in that direction.
The constant 3 in the upper bound on $cr$ comes from bounds on the maximum gap between primes. A proof of Cramér's conjecture on the size of prime gaps would imply that 3 can be replaced by any constant larger than~1, which in turn would lead to a smaller exponent in the polylogarithmic overhead.


\section{Proof of Theorem~\ref{thm:RAM}}

The data structure will choose either $\mathcal{A}_1$ or $\mathcal{A}_2$ of Corollary~\ref{cor:params}, or $\mathcal{A}_3$ of section~\ref{sec:small-radius} with size bounded in (\ref{eq:smallbound}), depending on which $i\in\{1,2,3\}$ minimizes $|\mathcal{A}_i| + \kappa_{\mathcal{A}_i}$. The term $n^{0.4/c}$ comes from part (3) of Corollary~\ref{cor:params} and the inequality $\ln(4) < 1.4$.

The resulting space usage is $\BO{|\mathcal{A}_i| n \log n + nd}$ bits, representing buckets by list of pointers to an array of all vectors in~$S$.
Also observe that the expected query time is bounded by $|\mathcal{A}_i| + \kappa_{\mathcal{A}_i}$.
\qed

\section{Conclusion and open problems}

We have seen that, at least in Hamming space, LSH-based similarity search can be implemented to avoid the problem of false negatives at little or no cost in efficiency compared to conventional LSH-based methods.
The methods presented are simple enough that they may be practical.
An obvious open problem is to completely close the gap, or show that a certain loss of efficiency is necessary (the non-constructive bound in section~\ref{sec:nofalseneg} shows that the gap is at most a factor $\BO{d}$).

It is of interest to investigate the possible time-space trade-offs.
CoveringLSH uses superlinear space and employs a data independent family of functions. 
Is it possible to achieve covering guarantees in linear or near-linear space?
Can data structures with very fast queries and polynomial space usage match the performance achievable with false negatives~\cite{DBLP:journals/corr/Laarhoven15a}?

Another interesting question is what results are possible in this direction for other spaces and distance measures, e.g., $\ell_1$, $\ell_2$, or $\ell_\infty$.
For example, a more practical alternative to the reduction of~\cite{DBLP:conf/stoc/Indyk07} for handling $\ell_1$ and $\ell_2$ would be interesting.

Finally, CoveringLSH is data independent. Is it possible to improve performance by using data dependent techniques?

\medskip

{\bf Acknowledgements.} The author would like to thank: Ilya Razenshteyn for useful comments; Thomas Dybdal Ahle, Ugo Vaccaro, and Annalisa De Bonis for providing references to work on explicit covering designs; Piotr Indyk for information on reduction from $\ell_1$ and $\ell_2$ metrics to the Hamming metric; members of the Scalable Similarity Search project for many rewarding discussions of this material.

\newpage

\bibliographystyle{ACM-Reference-Format-Journals}
\bibliography{coveringLSH}


\begin{thebibliography}{00}


\ifx \showCODEN    \undefined \def \showCODEN     #1{\unskip}     \fi
\ifx \showDOI      \undefined \def \showDOI       #1{{\tt DOI:}\penalty0{#1}\ }
  \fi
\ifx \showISBNx    \undefined \def \showISBNx     #1{\unskip}     \fi
\ifx \showISBNxiii \undefined \def \showISBNxiii  #1{\unskip}     \fi
\ifx \showISSN     \undefined \def \showISSN      #1{\unskip}     \fi
\ifx \showLCCN     \undefined \def \showLCCN      #1{\unskip}     \fi
\ifx \shownote     \undefined \def \shownote      #1{#1}          \fi
\ifx \showarticletitle \undefined \def \showarticletitle #1{#1}   \fi
\ifx \showURL      \undefined \def \showURL       #1{#1}          \fi

\bibitem[\protect\citeauthoryear{Alman and Williams}{Alman and
  Williams}{2015}]%
        {DBLP:conf/focs/AlmanW15}
{Josh Alman} {and} {Ryan Williams}. 2015.
\newblock \showarticletitle{Probabilistic Polynomials and Hamming Nearest
  Neighbors}. In {\em Proceedings of 56th Symposium on Foundations of Computer
  Science (FOCS)}. 136--150.
\newblock


\bibitem[\protect\citeauthoryear{Andoni, Indyk, Nguyen, and Razenshteyn}{Andoni
  et~al\mbox{.}}{2014}]%
        {andoni2014beyond}
{Alexandr Andoni}, {Piotr Indyk}, {Huy~L Nguyen}, {and} {Ilya Razenshteyn}.
  2014.
\newblock \showarticletitle{Beyond locality-sensitive hashing}. In {\em
  Proceedings of the 25th Symposium on Discrete Algorithms (SODA)}. 1018--1028.
\newblock


\bibitem[\protect\citeauthoryear{Andoni and Razenshteyn}{Andoni and
  Razenshteyn}{2015}]%
        {DBLP:conf/stoc/AndoniR15}
{Alexandr Andoni} {and} {Ilya Razenshteyn}. 2015.
\newblock \showarticletitle{Optimal Data-Dependent Hashing for Approximate Near
  Neighbors}. In {\em Proceedings of 47th Symposium on Theory of Computing
  (STOC)}. 793--801.
\newblock


\bibitem[\protect\citeauthoryear{Arasu, Ganti, and Kaushik}{Arasu
  et~al\mbox{.}}{2006}]%
        {Arasu_VLDB06}
{Arvind Arasu}, {Venkatesh Ganti}, {and} {Raghav Kaushik}. 2006.
\newblock \showarticletitle{Efficient Exact Set-Similarity Joins}. In {\em
  Proceedings of 32nd Conference on Very Large Data Bases (VLDB)}. 918--929.
\newblock


\bibitem[\protect\citeauthoryear{Cole, Gottlieb, and Lewenstein}{Cole
  et~al\mbox{.}}{2004}]%
        {Cole:2004:DMI:1007352.1007374}
{Richard Cole}, {Lee-Ad Gottlieb}, {and} {Moshe Lewenstein}. 2004.
\newblock \showarticletitle{Dictionary Matching and Indexing with Errors and
  Don't Cares}. In {\em Proceedings of 36th Symposium on Theory of Computing
  (STOC)}. ACM, 91--100.
\newblock


\bibitem[\protect\citeauthoryear{Deng, Li, Wen, and Feng}{Deng
  et~al\mbox{.}}{2015}]%
        {DBLP:journals/pvldb/DengLWF15}
{Dong Deng}, {Guoliang Li}, {He Wen}, {and} {Jianhua Feng}. 2015.
\newblock \showarticletitle{An Efficient Partition Based Method for Exact Set
  Similarity Joins}.
\newblock {\em {PVLDB}\/} {9}, 4 (2015), 360--371.
\newblock


\bibitem[\protect\citeauthoryear{Dudek}{Dudek}{2014}]%
        {dudek2014explicit}
{Adrian Dudek}. 2014.
\newblock \showarticletitle{An Explicit Result for Primes Between Cubes}.
\newblock {\em arXiv preprint arXiv:1401.4233\/} (2014).
\newblock


\bibitem[\protect\citeauthoryear{Gionis, Indyk, and Motwani}{Gionis
  et~al\mbox{.}}{1999}]%
        {gionis1999similarity}
{Aristides Gionis}, {Piotr Indyk}, {and} {Rajeev Motwani}. 1999.
\newblock \showarticletitle{Similarity search in high dimensions via hashing}.
  In {\em Proceedings of 25th Conference on Very Large Data Bases (VLDB)}.
  518--529.
\newblock


\bibitem[\protect\citeauthoryear{Gordon, Patashnik, and Kuperberg}{Gordon
  et~al\mbox{.}}{1995}]%
        {gordon1995new}
{Daniel~M Gordon}, {Oren Patashnik}, {and} {Greg Kuperberg}. 1995.
\newblock \showarticletitle{New constructions for covering designs}.
\newblock {\em Journal of Combinatorial Designs\/} {3}, 4 (1995), 269--284.
\newblock


\bibitem[\protect\citeauthoryear{Greene, Parnas, and Yao}{Greene
  et~al\mbox{.}}{1994}]%
        {greene1994multi}
{Dan Greene}, {Michal Parnas}, {and} {Frances Yao}. 1994.
\newblock \showarticletitle{Multi-index hashing for information retrieval}. In
  {\em Proceedings of 35th Symposium on Foundations of Computer Science
  (FOCS)}. 722--731.
\newblock


\bibitem[\protect\citeauthoryear{Hagerup}{Hagerup}{1998}]%
        {word-RAM}
{Torben Hagerup}. 1998.
\newblock \showarticletitle{Sorting and Searching on the Word {R}{A}{M}}.
\newblock In {\em Proceedings of 15th Symposium on Theoretical Aspects of
  Computer Science (STACS)}. Lecture Notes in Computer Science, Vol. 1373.
  Springer, 366--398.
\newblock


\bibitem[\protect\citeauthoryear{Hagerup, Miltersen, and Pagh}{Hagerup
  et~al\mbox{.}}{2001}]%
        {hagerup2001deterministic}
{Torben Hagerup}, {Peter~Bro Miltersen}, {and} {Rasmus Pagh}. 2001.
\newblock \showarticletitle{Deterministic dictionaries}.
\newblock {\em Journal of Algorithms\/} {41}, 1 (2001), 69--85.
\newblock


\bibitem[\protect\citeauthoryear{Har-Peled, Indyk, and Motwani}{Har-Peled
  et~al\mbox{.}}{2012}]%
        {HarPeled2012}
{Sariel Har-Peled}, {Piotr Indyk}, {and} {Rajeev Motwani}. 2012.
\newblock \showarticletitle{Approximate nearest neighbors: towards removing the
  curse of dimensionality}.
\newblock {\em Theory of Computing\/} {8}, 1 (2012), 321--350.
\newblock


\bibitem[\protect\citeauthoryear{Indyk}{Indyk}{2000}]%
        {DBLP:conf/soda/Indyk00}
{Piotr Indyk}. 2000.
\newblock \showarticletitle{Dimensionality reduction techniques for proximity
  problems}. In {\em Proceedings of Symposium on Discrete Algorithms (SODA)}.
  371--378.
\newblock


\bibitem[\protect\citeauthoryear{Indyk}{Indyk}{2007}]%
        {DBLP:conf/stoc/Indyk07}
{Piotr Indyk}. 2007.
\newblock \showarticletitle{Uncertainty principles, extractors, and explicit
  embeddings of l2 into l1}. In {\em Proceedings of 39th Symposium on Theory of
  Computing (STOC)}. 615--620.
\newblock


\bibitem[\protect\citeauthoryear{Indyk}{Indyk}{2015}]%
        {IndykPersonalCommunication2015}
{Piotr Indyk}. 2015.
\newblock Personal communication.  (2015).
\newblock


\bibitem[\protect\citeauthoryear{Indyk and Motwani}{Indyk and Motwani}{1998}]%
        {Indyk1998}
{Piotr Indyk} {and} {Rajeev Motwani}. 1998.
\newblock \showarticletitle{Approximate Nearest Neighbors: Towards Removing the
  Curse of Dimensionality}. In {\em Proceedings of 30th Symposium on the Theory
  of Computing (STOC)}. 604--613.
\newblock


\bibitem[\protect\citeauthoryear{Kapralov}{Kapralov}{2015}]%
        {kapralov2015smooth}
{Michael Kapralov}. 2015.
\newblock \showarticletitle{Smooth Tradeoffs between Insert and Query
  Complexity in Nearest Neighbor Search}. In {\em Proceedings of 34th Symposium
  on Principles of Database Systems (PODS)}. 329--342.
\newblock


\bibitem[\protect\citeauthoryear{Kucherov, No{\'e}, and Roytberg}{Kucherov
  et~al\mbox{.}}{2005}]%
        {kucherov2005multiseed}
{Gregory Kucherov}, {Laurent No{\'e}}, {and} {Mikhail Roytberg}. 2005.
\newblock \showarticletitle{Multiseed lossless filtration}.
\newblock {\em IEEE/ACM Transactions on Computational Biology and
  Bioinformatics (TCBB)\/} {2}, 1 (2005), 51--61.
\newblock


\bibitem[\protect\citeauthoryear{Kuzjurin}{Kuzjurin}{1995}]%
        {Kuzjurin:1995:DAG:204515.204520}
{Nikolai~N. Kuzjurin}. 1995.
\newblock \showarticletitle{On the Difference Between Asymptotically Good
  Packings and Coverings}.
\newblock {\em Eur. J. Comb.\/} {16}, 1 (Jan. 1995), 35--40.
\newblock
\showISSN{0195-6698}


\bibitem[\protect\citeauthoryear{Kuzjurin}{Kuzjurin}{2000}]%
        {kuzjurin2000explicit}
{Nikolai~N Kuzjurin}. 2000.
\newblock \showarticletitle{Explicit constructions of {R}{\"o}dl's
  asymptotically good packings and coverings}.
\newblock {\em Combinatorics, Probability and Computing\/} {9}, 03 (2000),
  265--276.
\newblock


\bibitem[\protect\citeauthoryear{Laarhoven}{Laarhoven}{2015}]%
        {DBLP:journals/corr/Laarhoven15a}
{Thijs Laarhoven}. 2015.
\newblock \showarticletitle{Tradeoffs for nearest neighbors on the sphere}.
\newblock {\em CoRR\/}  {abs/1511.07527} (2015).
\newblock


\bibitem[\protect\citeauthoryear{Minsky and Papert}{Minsky and Papert}{1987}]%
        {minsky1987perceptrons}
{Marvin~L Minsky} {and} {Seymour~A Papert}. 1987.
\newblock {\em Perceptrons - Expanded Edition: An Introduction to Computational
  Geometry}.
\newblock MIT press.
\newblock


\bibitem[\protect\citeauthoryear{Norouzi, Punjani, and Fleet}{Norouzi
  et~al\mbox{.}}{2012}]%
        {norouzi2012fast}
{Mohammad Norouzi}, {Ali Punjani}, {and} {David~J Fleet}. 2012.
\newblock \showarticletitle{Fast search in {H}amming space with multi-index
  hashing}. In {\em IEEE Conference on Computer Vision and Pattern Recognition
  (CVPR)}. 3108--3115.
\newblock


\bibitem[\protect\citeauthoryear{O’Donnell, Wu, and Zhou}{O’Donnell
  et~al\mbox{.}}{2014}]%
        {o2014optimal}
{Ryan O’Donnell}, {Yi Wu}, {and} {Yuan Zhou}. 2014.
\newblock \showarticletitle{Optimal lower bounds for locality-sensitive hashing
  (except when q is tiny)}.
\newblock {\em ACM Transactions on Computation Theory (TOCT)\/} {6}, 1 (2014),
  5.
\newblock


\bibitem[\protect\citeauthoryear{Panigrahy}{Panigrahy}{2006}]%
        {panigrahy2006entropy}
{Rina Panigrahy}. 2006.
\newblock \showarticletitle{Entropy based nearest neighbor search in high
  dimensions}. In {\em Proceedings of 17th Symposium on Discrete Algorithm
  (SODA)}. 1186--1195.
\newblock


\bibitem[\protect\citeauthoryear{Parnas}{Parnas}{2015}]%
        {ParnasPersonalCommunication}
{Michal Parnas}. 2015.
\newblock Personal communication.  (2015).
\newblock


\bibitem[\protect\citeauthoryear{Williams}{Williams}{2005}]%
        {williams2005new}
{Ryan Williams}. 2005.
\newblock \showarticletitle{A new algorithm for optimal 2-constraint
  satisfaction and its implications}.
\newblock {\em Theoretical Computer Science\/} {348}, 2 (2005), 357--365.
\newblock


\end{thebibliography}

\end{document}